\documentclass[preprint,showkeys,preprintnumbers,amsmath,amssymb,superscriptaddress]{revtex4}

\usepackage{amsthm,graphicx,amsfonts}

\usepackage{natbib}

\newcounter{lemmacounter}
\newtheorem{lemma}[lemmacounter]{Lemma}

\newcounter{theocounter}
\newtheorem{theorem}[theocounter]{Theorem}

\begin{document}

\title{R\'enyi entropy of the infinite well potential in momentum space and Dirichlet-like trigonometric functionals}

\author{A.I. Aptekarev}
\email{aptekaa@keldysh.ru}
\affiliation{Keldysh Institute for Applied Mathematics, Russian Academy of Sciences and Moscow State University, Moscow, Russia}

\author{J.S. Dehesa}
\email{dehesa@ugr.es}
\affiliation{Institute ``Carlos I'' for Computational and Theoretical Physics, University of Granada, Granada, Spain}
\affiliation{Department of Atomic, Molecular and Nuclear Physics, University of Granada, Granada, Spain}

\author{P. S\'anchez-Moreno}
\email{pablos@ugr.es}
\affiliation{Institute ``Carlos I'' for Computational and Theoretical Physics, University of Granada, Granada, Spain}
\affiliation{Department of Applied Mathematics, University of Granada, Granada, Spain}

\author{D.N. Tulyakov}
\email{dnt@mail.nnov.ru}
\affiliation{Keldysh Institute for Applied Mathematics, Russian Academy of Sciences and Moscow State University, Moscow, Russia}

\begin{abstract}
The momentum entropic moments and R\'enyi entropies of a
one-dimensional particle in an infinite well potential are found by means
of explicit calculations of some Dirichlet-like trigonometric integrals. The
associated spreading lengths and quantum uncertainty-like sums are also
provided.
\end{abstract}

%\pacs{03.65.Ta, 89.70.Cf}
%
%\ams{94A17, 62B10}


\maketitle

\section{Introduction}

A major goal of the information theory of quantum systems is the
determination of the uncertainty fundamental quantities  beyond the
well-known standard deviation (or its square, the variance) in both
position and momentum spaces \cite{ohya_04,uffink_thesis,frieden_04,hall_pra99}. 
They include the spreading measures of global (entropic moments,  the
R\'enyi, Shannon and Tsallis' entropies and lengths) and local (Fisher's
information) character \cite{renyi_70,shannon:bst48,tsallis_jsp88,frieden_04}. These information-theoretic measures quantify the
spread of the position and momentum one-particle densities of the system
(which are the basic variables according to the Density Functional Theory
\cite{parr_89}) in various complementary ways and, in contrast to the variance,
without reference to any specific point of the corresponding Hilbert space.
Moreover, they are closely related to various energetic and experimentally measurable quantities of the system \cite{frieden_04,parr_89,dehesa_ijqc10,dehesa:pra94,bialynicki_pra06}.

However, they have not yet been or cannot be analytically calculated; not
even for the one-dimensional systems with a Coulomb, oscillator or
infinite well
potential despite they are very often used to explain and predict
numerous physico-mathematical phenomena in science and technology. The
present knowledge of the information-theoretic measures of the hydrogenic,
oscillator-like and particle-in-a-box systems is described in Refs.
\cite{dehesa_ijqc10}, \cite{dehesa:jmp98,dehesa:pra94}
and \cite{lopezrosa_jmc11}, respectively. Therein we realize that the entropic moments
and the R\'enyi, Shannon and Tsallis entropies are not yet known
except for the lowest- and highest-lying states, mainly because they are
power (entropic moments, R\'enyi and Tsallis entropies) or logarithmic
(Shannon entropy) functionals of the corresponding densities. The
calculation of the Fisher information is somewhat easier because of its
close connection with the kinetic energy.

In this paper we will center around the particle-in-a-box system, which is
composed by a particle moving in the infinite well potential $V(x) = 0$
for $|x| < a$ and $\infty$ otherwise. This canonical system has been used
to simplify the description of numerous scientific phenomena in nuclei
\cite{bohr_98}, polymers \cite{pederson:prb00} and various nanosystems
\cite{rubio:prl99,harrison_05} as well as in chaos \cite{hu:prl99}, among
others.
The quantum-mechanical states of this system
are characterized \cite{majernik:jpa99,majernik:jpa97,lopezrosa_jmc11} by the energies
\[
E_n=\frac{\pi^2}{8a^2}n^2; \; n=1,2,3,\ldots,
\]
and the densities
\[
\rho_n(x)=\left\{
\begin{array}{ll}
  \displaystyle \frac{1}{a} \sin^2\left[\frac{\pi n}{2a}(x-a) \right]; & |x|\le a\\[3mm]
  0; & |x|>a
\end{array}
\right.
\]
and
\begin{equation}
  \gamma_n(p)=
  \frac{\pi n^2}{2}\frac{\sin^2\left[ap -\frac{\pi n}{2} \right]}{\left(
  a^2 p^2-\frac{\pi^2 n^2}{4} \right)^2}; \; p\in(-\infty,+\infty),
  \label{eq:density_momentum}
\end{equation}
in position and momentum spaces, respectively. Recently, the power moments
$\langle x^k\rangle$
and the entropic moments $W_k[\rho_n] = \langle \rho_n^{k-1}\rangle$ 
in position space have been
determined \cite{sanchezruiz:pla97,sanchezruiz:jpa99,lopezrosa_jmc11} 
what has allowed us to find the values for the
position variance, R\'enyi, Shannon and Tsallis entropies; in addition the
position Fisher information $F[\rho_n]$ is also known
\cite{lopezrosa_jmc11}.  
In momentum space,
however, the situation is very different; while the second-order
moment $\langle p^2\rangle$
and the Fisher information $F[\gamma_n]$ have known values
\cite{lopezrosa_jmc11}, the
entropic moments and consequently the R\'enyi and Tsallis entropies have not
yet been able to be calculated, mainly because all these three quantities
depend \cite{lopezrosa_jmc11} on the Dirichlet-like trigonometric integrals
\cite{dym_72}
\begin{equation}
  I_{n,k}=\int_{-\infty}^{+\infty}
  \left[ \frac{\sin^2\left(t-\frac{\pi n}{2} \right)}{\left( t^2
  -\frac{\pi^2 n^2}{4} \right)^2}
  \right]^k dt,
  \label{eq:integral_ink}
\end{equation}
(compare with (\ref{eq:density_momentum})).

This functional has been recently computed only in the two following cases
\cite{lopezrosa_jmc11}: (a)
$n$ fixed and the first few values of $k$, and (b) $k$ fixed and very large
$n$. A
main purpose of this paper is to calculate this functional for the generic
pair $(n,k)$, and then to find the values of all entropic moments, the R\'enyi
entropies and lengths, the Tsallis entropies and the associated uncertainty
relations in momentum space.
This opens the way to calculate various complexity measures not only in position space but also in momentum space. This is the case of the L\'opez-Ruiz-Mancini-Calbet (LMC, in short) \cite{catalan:pre02} and the Fisher-R\'enyi \cite{romera_pla08} complexities, where the R\'enyi entropy plays an important role; in particular, the LMC complexity of the particle-in-a-box  has been numerically studied in \cite{lopezruiz:osid09,nagy:pla09} and mathematically considered in \cite{lopezrosa_jmc11}. It has been recently reviewed the high relevance of these statistical complexity measures in analysing a great diversity of physical phenomena related with the internal disorder of the intrinsic structure of many-electron systems \cite{sen_11}.

The paper is structured as follows. Firstly, in Section
\ref{sec:measures} the
information-theoretic measures of a one-dimensional probability
distribution needed in this work are briefly discussed. Then, in Section
\ref{sec:trigonometric_integral},
the calculation of the trigonometric integral $I_{n,k}$  is explicitly carried
out. The resulting expression is used in Section \ref{sec:calculations} to find the values of
the entropic moments, R\'enyi's measures and Tsallis'
entropies in momentum space, as well as the associated uncertainty-like
expressions
which
combines
them with the corresponding quantities in position space with integer orders. Finally, the conclusions and some open problems are given.

\section{Information-theoretic measures of a probability density: Basics}
\label{sec:measures}

The uncertainty measures of a random variable $X$ are given by the
spreading measures of the corresponding probability density $\rho(x)$,
$x\in(-\infty,+\infty)$. These measures are defined either in terms of the
moments-around-the-origin (or simply, power moments) $\langle x^k\rangle$
(such as, e.g., the variance $V[\rho]=\langle x^2\rangle-\langle
x\rangle^2$) or in terms of the frequency or entropic moments $W_k[\rho]$ of
$\rho(x)$, which are defined by
\[
\langle x^k\rangle=\int_{-\infty}^{+\infty} x^k \rho(x) dx,
\]
and
\begin{equation}
W_k[\rho]=\left\langle \rho^{k-1}\right\rangle=
\int_{-\infty}^{+\infty} \left[ \rho(x) \right]^k dx,
\label{eq:entropic_moment_definition}
\end{equation}
respectively. From the latter quantities various information-theoretic
measures have been defined, such as the R\'enyi entropy \cite{renyi_70}
\begin{equation}
R_\alpha[\rho]=\frac{1}{1-\alpha}\ln W_\alpha[\rho]
=\frac{1}{1-\alpha}\ln\left\langle \rho^{\alpha-1}\right \rangle,\,\alpha>0,
\label{eq:renyi_entropy_definition}
\end{equation}
and the Tsallis entropy \cite{tsallis_jsp88}
\begin{equation}
  T_\alpha[\rho]=\frac{1}{\alpha-1}\left( 1-W_\alpha[\rho] \right)
  =\frac{1}{\alpha-1}\left( 1-\left\langle \rho^{\alpha-1}\right\rangle
  \right),\,\alpha>0.
  \label{eq:tsallis_entropy_definition}
\end{equation}
The limiting case $\alpha\to 1$ of these two quantities is the celebrated
Shannon entropy \cite{shannon:bst48}
\[
  S[\rho]=-\int_{-\infty}^{+\infty} \rho(x)\ln\rho(x) dx.
\]

In contrast to these three entropy quantities which measure in various (but
complementary) ways the total extent in which the random variable is
distributed, there exists a qualitatively different spreading quantity: the
Fisher information with respect to the location parameter (or simply Fisher information, heretoforth), defined \cite{frieden_04} by
\[
  F[\rho]=\left\langle \left[ \frac{d}{dx}\ln\rho(x)
  \right]^2\right\rangle=
  \int_{-\infty}^{+\infty} \frac{\left[ \rho'(x) \right]^2}{\rho(x)}dx.
\]
Indeed, contrary to the variance and the R\'enyi, Shannon and Tsallis
entropies, the Fisher information has a locality property since its value
mainly comes from the regions where the density is more oscillatory (i.e. when the
density has more nodes per unit argument of $x$). In other terms, the
Fisher information is very sensitive to the fluctuations of $\rho(x)$.
Then, it provides an estimation of the oscillatory character of the
density while the R\'enyi, Shannon and Tsallis entropies measure in different ways the total extent in which the density is distributed.

The five spreading measures previously defined, although interesting
\textit{per se}, 
cannot be mutually compared because either they are  dimensionless or they do not have the same units. To
overcome this problem, following Hall \cite{hall_pra99}, we will use instead the
so-called ``direct spreading measures'' to analyze the uncertainty of the
random variable; namely, the standard deviation, $\Delta x=\left( V[\rho]
\right)^\frac12$, and the information-theoretic lengths of R\'enyi, Shannon (also called entropy power)
and Fisher defined by
\[
  L_\alpha^R[\rho]\equiv\exp\left( R_\alpha[\rho] \right),
\]
\[
  L^S[\rho]\equiv \exp\left( S[\rho] \right),
\]
and
\[
  \delta x\equiv L^F[\rho]\equiv \frac{1}{\sqrt{F[\rho]}},
\]
respectively. All these quantities have the same units as the random
variable, and they satisfy the three following properties: translation and
reflection invariance, linear scaling and vanishing as the density approach
to the Dirac delta. Moreover, all of them fulfil an uncertainty relation.

\section{Calculation of the Dirichlet-like trigonometric functionals}
\label{sec:trigonometric_integral}

In this Section we give and prove Theorem \ref{th:mainTheorem}, from which follows the main result of the paper.

\begin{theorem}
  The integral $I_{n,k}$, defined in (\ref{eq:integral_ink}), has the value
  \[
 I_{n,k}=(\pi n)^{-2k} \sum_{j=0}^{k-1} 
  \left(
  \begin{array}{c}
  2j+2k-1\\
  2k-1
  \end{array}
  \right)
  \frac{2\pi (\pi n)^{-2j} \left( -\frac14 \right)^j}{(2k-2j-1)!}
  \sum_{i=0}^{k-1} (-1)^i 
    \left(
    \begin{array}{c}
    2k\\
    i
    \end{array}
    \right)
  (k-i)^{2k-2j-1},
  \]
  for $n=1,2\ldots$, and $k=1,2,\ldots$
  \label{th:mainTheorem}
\end{theorem}

In particular for $k=1$ and $k=2$, one easily has the values
\begin{equation}
I_{n,1}=\frac{2}{\pi n^2};\quad I_{n,2}=\frac{4}{3\pi^3n^4}\left(
1+\frac{15}{2\pi^2n^2}
\right),
\label{eq:in1_in2}
\end{equation}
in agreement with the corresponding values already obtained in
\cite{lopezrosa_jmc11}. For
completeness, let us also collect here that the asymptotical $(n\to\infty)$
values of $I_{n,k}$ are known \cite{lopezrosa_jmc11} to be
\begin{equation}
  I_{n,k}
  \simeq
  \frac{b_k}{\pi^{2k-1}n^{2k}},
  \label{eq:ink_asymptotics}
\end{equation}
where
\begin{equation}
  b_k=4k\sum_{i=0}^{k-1} \frac{(-1)^i (k-i)^{2k-1}}{i!(2k-i)!};\;k\ge 1.
  \label{eq:bk}
\end{equation}
It is worth pointing out that Eq. (\ref{eq:ink_asymptotics}) can be obtained from Theorem \ref{th:mainTheorem} simply by taking the first term of the sum (i.e., $j=0$).

To prove Theorem \ref{th:mainTheorem}, we will use three Lemmas. 
In Lemma \ref{lemma1} we expand the
reciprocal of the product $(t-\pi n/2)^{2k}(t+\pi n/2)^{2k}$
contained in the Dirichlet-like kernel \cite{dym_72} of
the integral (3) as a sum of simple fractions. Putting this result into
Eq.(\ref{eq:integral_ink}), Lemma \ref{lemma2} shows that the calculation of 
$I_{n,k}$  reduces to the evaluation
of the integral 
\begin{equation}
K_{n,j}=\int_{-\infty}^\infty \frac{\sin^{2k}u}{u^{2(k-j)}}du.
\label{eq:integral_knj}
\end{equation}
Then, in Lemma \ref{lemma3} this integral is explicitly calculated, and
consequently Theorem \ref{th:mainTheorem} follows in a straightforward
manner.

\begin{lemma} 
The following identity holds true:
  \begin{eqnarray}
   \frac{1}{\left( t-\frac{\pi n}{2} \right)^{2k}
  \left( t+\frac{\pi n}{2} \right)^{2k}}
  =(\pi n)^{-2k} \sum_{j=0}^{2k-1} (\pi n)^{-j}
    \left(
    \begin{array}{c}
    j+2k-1\\
    2k-1
    \end{array}
    \right)\nonumber\\
\times  \left[ \left( t+\frac{\pi n}{2} \right)^{j-2k}
  +(-1)^j \left( t-\frac{\pi n}{2} \right)^{j-2k}\right].
  \end{eqnarray}
  \label{lemma1}
\end{lemma}

\begin{proof}[Proof of Lemma \ref{lemma1}]
  We will expand $\frac{1}{\left( t-\frac{\pi n}{2} \right)^{2k}
  \left( t+\frac{\pi n}{2} \right)^{2k}}:=A(t)$ into a sum of simple fractions;
  i.e. each term of this sum has one pole in variable $t$ with certain multiplicity and the numerators (residues) of the terms do not depend on $t$.
  First, let us take
  \[
 u:=t-\frac{\pi n}{2} \Rightarrow
  A=\frac{(u+\pi n)^{-2k}}{u^{2k}}=(\pi n)^{-2k}
  \left[ \sum_{j=0}^{2k-1} u^{j-2k} (\pi n)^{-j}
    \left(
    \begin{array}{c}
    -2k\\
    j
    \end{array}
    \right)\right]+f(u),
  \]
  where $f(u)=O(1)$ when $u\to 0$. Now, let us use
  \[
  u:=t+\frac{\pi n}{2} \Rightarrow
  A=\frac{(u-\pi n)^{-2k}}{u^{2k}}=(\pi n)^{-2k}
  \left[ \sum_{j=0}^{2k-1} (-1)^j u^{j-2k} (\pi n)^{-j}
    \left(
    \begin{array}{c}
    -2k\\
    j
    \end{array}
    \right)
  \right]+\tilde{f}(u),
  \]
  where $\tilde{f}(u)=O(1)$ when $u\to 0$.
  Here we used the notation
  \[
    \left(
    \begin{array}{c}
    -2k\\
    j
    \end{array}
    \right)
  =\frac{(-2k)(-2k-1)\cdots(-2k-j+1)}{j!}.
  \]

Then, coming back to variable $t$ we have the following representation for $A(t)$ as sum of simple fractions:
\[
 A(t)=(\pi n)^{-2k} \sum_{j=0}^{2k-1} (\pi n)^{-j}
   \left(
      \begin{array}{c}
      -2k\\
      j
      \end{array}
      \right)
  \left[ \left( t-\frac{\pi n}{2} \right)^{j-2k} +(-1)^j
  \left( t+\frac{\pi n}{2} \right)^{j-2k}\right]+  \hat{f}(t),
  \]
where $\hat{f}(t)=f(t)+\tilde{f}(t)$. Since poles of $A(t)$ with their multiplicities coincide with the poles of the sum from the right hand side, then function $\hat{f}(t)$ has no poles in the whole complex plain and therefore, by Liouville theorem, function $\hat{f}(t)$ must be a constant. Taking into account that $A(\infty)=0$ and the sum from the right hand side also equal zero for $t\to\infty$ we get $\hat{f}(t)=0$.

  Finally, passing from $
 \left(
    \begin{array}{c}
    -2k\\
    j
    \end{array}
    \right)$ to $
     \left(
        \begin{array}{c}
        j+2k-1\\
        2k-1
        \end{array}
     \right)$ we
  arrive to validity of Lemma \ref{lemma1}.
\end{proof}

\begin{lemma}
The integrals $I_{n,k}$ defined by (\ref{eq:integral_ink}) simplify as
  \[
  I_{n,k}=(\pi n)^{-2k}
  \sum_{j=0}^{k-1} 2 (\pi n)^{-2j}
  \left(
          \begin{array}{c}
          2j+2k-1\\
          2k-1
          \end{array}
       \right)
  \int_{-\infty}^\infty
  \frac{\sin^{2k}u}{u^{2(k-j)}}du.
  \]
  \label{lemma2}
\end{lemma}

\begin{proof}[Proof of Lemma \ref{lemma2}]
We substitute the sum from Lemma \ref{lemma1} in the expression (\ref{eq:integral_ink}) of $I_{n,k}$. Then, it is straightforward to see that
\begin{eqnarray}
I_{n,k}=(\pi n)^{-2k}
  \sum_{j=0}^{2k-1} (\pi n)^{-j}
\left(
\begin{array}{c}
j+2k-1\\
2k-1
\end{array}
\right) 
  \left[
  \int_{-\infty}^{\infty} \sin^{2k}\left(t-\frac{\pi n}{2}\right) 
\left( t+\frac{\pi n}{2} \right)^{j-2k}dt\right.\nonumber\\
  \left.+(-1)^j\int_{-\infty}^{\infty} \sin^{2k}\left(t-\frac{\pi n}{2}\right)  \left( t-\frac{\pi n}{2} \right)^{j-2k}dt\right].
  \label{eq:ink_proof_lemma2}
\end{eqnarray}
Due to the periodicity of the sine function, these two integrals have the same value:
\begin{eqnarray*}
\int_{-\infty}^{\infty} \sin^{2k}\left(t-\frac{\pi n}{2}\right)  \left( t-\frac{\pi n}{2} \right)^{j-2k}dt
=
\int_{-\infty}^{\infty} \sin^{2k}\left(z+\frac{\pi n}{2}\right)  \left( z+\frac{\pi n}{2} \right)^{j-2k}dz\\
=
\int_{-\infty}^{\infty} \sin^{2k}\left(z-\frac{\pi n}{2}\right)  \left( z+\frac{\pi n}{2} \right)^{j-2k}dz,
\end{eqnarray*}
where the change of variable $t = z+\pi n$ has been used in the first equality.

Thus, the terms in (\ref{eq:ink_proof_lemma2}) with odd $j$ vanish, so we can express $I_{n,k}$ as
\[
I_{n,k}=(\pi n)^{-2k}
  \sum_{j=0}^{k-1} (\pi n)^{-2j}
\left(
\begin{array}{c}
2j+2k-1\\
2k-1
\end{array}
\right)
  2
  \int_{-\infty}^{\infty} \sin^{2k}\left(t-\frac{\pi n}{2}\right) 
\left( t-\frac{\pi n}{2} \right)^{2j-2k}dt.
\]
Finally, the change of variable $u=t-\frac{\pi n}{2}$ yields the statement of the lemma.

\end{proof}

\begin{lemma} 
The integrals $K_{n,j}$ defined by (\ref{eq:integral_knj}) are given by
  \[
  \int_{-\infty}^\infty \frac{\sin^{2k} u}{u^{2(k-j)}}du
  =
  \frac{\pi \left( -\frac14 \right)^j}{(2k-2j-1)!}
  \sum_{i=0}^{k-1} (-1)^i
\left(
\begin{array}{c}
2k\\
i
\end{array}
\right)
  (k-i)^{2k-2j-1}.
  \]
  \label{lemma3}
\end{lemma}

\begin{proof}[Proof of Lemma \ref{lemma3}]
  We use the Fourier transform
  \[
  \mathcal{F}[f](\omega):=\int_{-\infty}^\infty f(u) e^{-i\omega u}du,
  \]
  to compute integral
  \[
  \int_{-\infty}^\infty f(u)du=\mathcal{F}[f](0),
  \quad f(u):=\left( \frac{\sin u}{u} \right)^{2(k-j)} \sin u^{2j}.
  \]
  We have
  \begin{eqnarray*}
\mathcal{F}\left[ \frac{\sin u}{u} \right](\omega)
  =
  \int_{-\infty}^\infty \frac{e^{iu}-e^{-iu}}{2iu}e^{-i\omega u} du
  =
  \int_{+i\infty}^{-i\infty} \frac{e^{-s}-e^s}{2s} e^{\omega s}
  \frac{ds}{i}\\
  =
  \frac{1}{2\pi i}\int_{-i\infty}^{+i\infty}
  \left( \pi \frac{e^s-e^{-s}}{s} \right)e^{\omega s} ds,
  \end{eqnarray*}
  where $s=-iu$. The last integral is computed using the table of inverse
  Laplace transforms (see \cite{abramowitz_64}, formulas 29.2.2 and 29.3.61):
  \[
  \mathcal{F}\left[ \frac{\sin u}{u} \right](\omega) 
  =
  \pi \left( H(\omega+1) - H(\omega-1) \right),
  \]
  where $H(\omega)$ is the Heaviside function
  \begin{eqnarray*}
  H(\omega-k)=
  \left\{
  \begin{array}{ll}
    1, & \omega\ge k\\
    0, & \omega<k
  \end{array}
  \right..
  \end{eqnarray*}
  Then we have
  \[
  \mathcal{F}[\sin u](\omega) = \frac{1}{i} \frac{d}{d\omega}
  \mathcal{F}\left[ \frac{\sin u}{u} \right](\omega)
  =
  \frac{\pi}{i}\left[ \delta(\omega-1)-\delta(\omega+1) \right],
  \]
  where $\delta(\omega)$ is a Dirac delta function, which is characterized
  by
  \begin{equation}
  F(\omega)*\delta(\omega-a)=\frac{1}{2\pi}F(\omega-a),
  \label{eq:delta_convolution}
  \end{equation}
  where the convolution for Fourier transforms is defined as
  \[
  F(\omega) * G(\omega) =\frac{1}{2\pi}
  \int_{-\infty}^\infty F(\xi)G(\omega-\xi)d\xi,
  \]
  and we have
  \[
  \mathcal{F}[f](\omega) * \mathcal{F}[g](\omega)=\mathcal{F}[fg](\omega).
  \]
  Therefore, since (\ref{eq:delta_convolution})
  \[
  H(\omega+1)-H(\omega-1)=H(\omega)*
  \left[ \delta(\omega+1)-\delta(\omega-1) \right],
  \]
  we obtain
  \begin{eqnarray}
  \mathcal{F}\left[ \left( \frac{\sin u}{u} \right)^{2(k-j)} \sin^{2j} u
    \right](\omega)
    =(-1)^j \pi^{2k} \underbrace{H*H*\cdots *H}_{2(k-j)} \nonumber\\
    \,*\,\underbrace{\left[ \delta(\omega+1)-\delta(\omega-1) \right]*
    \cdots * \left[ \delta(\omega+1)-\delta(\omega-1) \right]}_{2k}.
    \label{eq:fourier_transform}
  \end{eqnarray}
  Computing
  \[
  \underbrace{H*\cdots *H}_{2(k-j)}
  =
  \frac{1}{(2\pi)^{2(k-j)-1}}
  \frac{|\omega_+|^{2(k-j)-1}}{(2(k-j)-1)!},\quad
  |\omega_+|:=\omega H(\omega),
  \]
  we substitute the result in (\ref{eq:fourier_transform}), and using
  (\ref{eq:delta_convolution}) we can obtain an explicit form for
  $\mathcal{F}\left[ \left( \frac{\sin u}{u} \right)^{2(k-j)} \sin^{2j} u
  \right](\omega)$. Finally, taking the obtained expression for $\omega=0$,
  we arrive to assertion of the Lemma.

\end{proof}

\section{R\'enyi measures and Tsallis entropies in momentum space}
\label{sec:calculations}

In this Section we give the values of the entropic moments and the R\'enyi
and Tsallis entropies as well as the associated
R\'enyi lengths of the infinite-well potential in momentum space. Moreover, the uncertainty-like
expressions
which
combines
these momentum quantities with the corresponding ones in position space are also provided for integer orders.

According to Eqs. (\ref{eq:density_momentum}) and
(\ref{eq:entropic_moment_definition}), one has that the momentum entropic
moments are given by
\begin{equation}
  W_k[\gamma_n]=\int_{-\infty}^{+\infty} \left[ \gamma_n(p) \right]^k dp=
  \left( \frac{\pi n^2}{2} \right)^k a^{k-1} I_{n,k},\quad k=1,2,\ldots,
  \label{eq:entropic_moment_momentum}
\end{equation}
where the integral $I_{n,k}$ is given by Theorem \ref{th:mainTheorem}.
Notice that taken into account (\ref{eq:in1_in2}), one has the normalization
$W_1[\gamma_n]=1$ and the averaging momentum density
\[
  W_2[\gamma_n]=\langle \gamma_n\rangle=\frac{a}{3\pi}\left(
  1+\frac{15}{2\pi^2 n^2}
  \right),
\]
for $k=1$ and $2$, respectively. Moreover, taking into account
(\ref{eq:ink_asymptotics}) one has the asymptotical values
\[
  W_k[\gamma_n]\simeq \frac{b_k}{2^k \pi^{k-1}}a^{k-1};\; k\ge 1,\;
  n\to+\infty.
\]

Now, from Eqs. (\ref{eq:density_momentum}),
(\ref{eq:renyi_entropy_definition}) and (\ref{eq:entropic_moment_momentum})
one finds the values for the momentum R\'enyi entropies and lengths
\begin{equation}
  R_k[\gamma_n]=\frac{1}{1-k}\ln W_k[\gamma_n]=
  -\ln a + \frac{1}{1-k}\ln \left( \left( \frac{\pi n^2}{2} \right)^k
  I_{n,k} \right)
  \label{eq:renyi_entropy_momentum}
\end{equation}
and
\begin{equation}
  L_k^R[\gamma_n]=\exp\left( R_k[\gamma_n] \right)=\left[ W_k[\gamma_n]
  \right]^{\frac{1}{1-k}}=\frac{1}{a} \left( \frac{\pi n^2}{2}
  \right)^{\frac{k}{1-k}} \left( I_{n,k} \right)^{\frac{1}{1-k}}
  \label{eq:renyi_length_momentum},
\end{equation}
respectively. Moreover, from Eqs. (\ref{eq:density_momentum}),
(\ref{eq:tsallis_entropy_definition}) and
(\ref{eq:entropic_moment_momentum}) one obtains the values for the momentum
Tsallis entropies
\begin{equation}
T_k[\gamma_n]=\frac{1}{k-1}\left( 1-W_k[\gamma_n] \right)=
\frac{1}{k-1}\left[ 1-\left( \frac{\pi n^2}{2} \right)^k a^{k-1} I_{n,k}
\right].
\label{eq:tsallis_entropy_momentum}
\end{equation}

The values of $I_{n,k}$ given by Theorem \ref{th:mainTheorem} together with Eqs. (\ref{eq:renyi_entropy_momentum})-(\ref{eq:tsallis_entropy_momentum}) provide the momentum R\'enyi and Tsallis quantities, respectively, in a straightforward manner. Moreover, since the R\'enyi and Tsallis entropies in position space are known \cite{sanchezruiz:jpa99,lopezrosa_jmc11} to be
\[
R_k[\rho_n]=\ln a +\frac{1}{1-k}\ln\left(\frac{2\Gamma\left(k+\frac12\right)}{\sqrt{\pi}\Gamma(k+1)}\right)
\]
and
\[
T_k[\rho_n]=\frac{1}{k-1}\left[1-\frac{2}{a^{k-1}\sqrt{\pi}}\frac{\Gamma\left(k+\frac12\right)}{\Gamma(k+1)}\right],
\]
respectively, with $k=1,2,\ldots$, it is possible to easily calculate the position-momentum R\'enyi uncertainty-like sum $R_k[\rho_n]+R_l[\gamma_n]$ and Tsallis uncertainty-like quotient $\left[1+(1-k)T_k[\rho_n]\right]^\frac{1}{2k}\left[1+(1-l)T_l[\gamma_n]\right]^{-\frac{1}{2l}}$.

Finally, let us examine the position-momentum product of the R\'enyi lengths, $L_k^R[\rho_n]\times L_k^R[\gamma_n]$. From (\ref{eq:renyi_length_momentum}) and the position R\'enyi length
\[
L_k^R[\rho_n]=2^{2+\frac{1}{k-1}} a
\left(
\begin{array}{c}
2k\\
k
\end{array}
\right)^{-\frac{1}{k-1}},
\]
one has
\[
L_k^R[\rho_n]\times L_l^R[\gamma_n]=2^{2+\frac{1}{k-1}}
\left(
\begin{array}{c}
2k\\
k
\end{array}
\right)^{-\frac{1}{k-1}} \left[\left(\frac{\pi n^2}{2}\right)^l I_{n,l}\right]^{-\frac{1}{l-1}}.
\]
Now the use of Theorem \ref{th:mainTheorem} in this expression yields immediately the explicit value of this position-momentum uncertainty-like product. For further discussions on the uncertainty measures of the particle-in-a-box system, and its generalization to $D$-dimensions we refer to Ref. \cite{lopezrosa_jmc11}.

\section{Conclusions and open problems}

In this paper we have extended the information-theoretic analysis of the non-relativistic particle-in-a-box system (i.e., a particle moving in an infinite-well potential) \cite{lopezrosa_jmc11,majernik:jpa99,majernik:jpa97,sanchezruiz:pla97,sanchezruiz:jpa99}, by explicitly calculating the R\'enyi and Tsallis entropies of integer order for all ground and excited quantum states in momentum space.
It has required the exact evaluation of the trigonometric functionals $I_{n,k}$ (see Eq. (\ref{eq:integral_ink})) with the Dirichlet-like kernel $\left(\sin\theta_n/\theta_n\right)^{2k}$, so much useful in Fourier analysis.
It is worth pointing out here that the calculation of these quantities for non-integer orders in both position and momentum spaces remains to be an open problem; they are needed to set up the corresponding uncertainty relations
which link position and momentum quantities with conjugated orders
\cite{zozor_pa08,bialynicki_pra06,maassen_prl88,rajagopal_pla95,uffink_thesis}.
Moreover, the evaluation of the momentum Shannon entropy is not yet known except for the lowest and highest excited states of the system. This is because this momentum functional has a trigonometric kernel of logarithmic form (see \cite[Eq. (13)-(14)]{lopezrosa_jmc11}).
Finally let us point out that the generalization to $D$ dimensions and the inclusion of the relativistic effects to the entropies of the infinite-well potential are further interesting open problems.

\section*{Acknowledgments}

AIA is partially supported by the grants Chair Excellence program of Universidad Carlos III Madrid, Spain and Bank Santander, and grant RFBR – 11-01-00245.
DNT acknowledges support to the grant RFBR – 10-01-000682.
JSD and PSM are very grateful for partial support to Junta de Andaluc\'{\i}a
(under grants FQM-4643 and FQM-2445) and Ministerio de Ciencia e Innovaci\'on under project FIS2011-24540. JSD and PSM belong to the Andalusian research group FQM-0207. We appreciate C. Vignat's help in deriving the asymptotics (\ref{eq:ink_asymptotics})-(\ref{eq:bk}).


\begin{thebibliography}{0}
\expandafter\ifx\csname natexlab\endcsname\relax\def\natexlab#1{#1}\fi
\expandafter\ifx\csname bibnamefont\endcsname\relax
  \def\bibnamefont#1{#1}\fi
\expandafter\ifx\csname bibfnamefont\endcsname\relax
  \def\bibfnamefont#1{#1}\fi
\expandafter\ifx\csname citenamefont\endcsname\relax
  \def\citenamefont#1{#1}\fi
\expandafter\ifx\csname url\endcsname\relax
  \def\url#1{\texttt{#1}}\fi
\expandafter\ifx\csname urlprefix\endcsname\relax\def\urlprefix{URL }\fi
\providecommand{\bibinfo}[2]{#2}
\providecommand{\eprint}[2][]{\url{#2}}

\end{thebibliography}


\begin{thebibliography}{10}

\bibitem{ohya_04}
M.~Ohya and D.~Petz.
\newblock {\em Quantum Entropy and Its Use}.
\newblock Springer, Berlin, 2004.

\bibitem{uffink_thesis}
J.B.M. Uffink.
\newblock {\em Measures of {U}ncertainty and the {U}ncertainty {P}rinciple}.
\newblock PhD Thesis, University of Utrecht, 1990.
\newblock See also references herein.

\bibitem{frieden_04}
B.~R. Frieden.
\newblock {\em Science from {F}isher {I}nformation}.
\newblock Cambridge University Press, Cambridge, 2004.

\bibitem{hall_pra99}
M.~J.~W. Hall.
\newblock Universal geometric approach to uncertainty, entropy and information.
\newblock {\em Phys. Rev. A}, 59:2602--2615, 1999.

\bibitem{renyi_70}
A.~R\'enyi.
\newblock {\em Probability Theory}.
\newblock North Holland, Amsterdam, 1970.

\bibitem{shannon:bst48}
C.~E. Shannon.
\newblock A mathematical theory of communication.
\newblock {\em Bell Syst. Tech. J.}, 27:379--423 and 623--656, 1948.

\bibitem{tsallis_jsp88}
C.~Tsallis.
\newblock Possible generalization of {B}oltzmann-{G}ibbs statistics.
\newblock {\em J. Stat. Phys.}, 52:479--487, 1988.

\bibitem{parr_89}
R.~G. Parr and W.~Yang.
\newblock {\em Density-{F}unctional {T}heory of {A}toms and {M}olecules}.
\newblock Oxford Univ. Press, New York, 1989.

\bibitem{dehesa_ijqc10}
J.~S. Dehesa, S.~L\'opez-Rosa, A.~Mart\'{\i}nez-Finkelshtein, and R.~J.
  Y\'a{\~n}ez.
\newblock Information theory of $d$-dimensional hydrogenic systems: Application
  to circular and {R}ydberg states.
\newblock {\em Int. J. Quantum Chem.}, 110:1529--1548, 2010.

\bibitem{dehesa:pra94}
J.~S. Dehesa, W.~van Assche, and R.~J. Y\'{a}{\~n}ez.
\newblock Position and momentum information entropies of the {D}-dimensional
  harmonic oscillator and hydrogen atom.
\newblock {\em Phys. Rev. A}, 50:3065--3079, 1994.

\bibitem{bialynicki_pra06}
I.~Bialynicki-Birula.
\newblock Formulations of uncertainty relations in terms of {R}\'enyi
  entropies.
\newblock {\em Phys. Rev. A}, 74:052101, 2006.

\bibitem{dehesa:jmp98}
J.~S. Dehesa, R.~J. Y\'{a}{\~n}ez, A.~I. Aptekarev, and V.~S. Buyarov.
\newblock Strong asymptotics of {L}aguerre polynomials and information
  entropies of two-dimensional harmonic oscillator and one-dimensional
  {C}oulomb potentials.
\newblock {\em J. Math. Phys.}, 39:3050--3060, 1998.

\bibitem{lopezrosa_jmc11}
S.~L\'opez-Rosa, J.~Montero, P.~S\'anchez-Moreno, J.~Venegas, and J.~S. Dehesa.
\newblock Position and momentum information-theoretic measures of a
  $d$-dimensional particle-in-a-box.
\newblock {\em J. Math. Chem.}, 49:971--994, 2011.

\bibitem{bohr_98}
A.~Bohr and B.~R. Mottelson.
\newblock {\em Nuclear Structure}.
\newblock World Scientific, Singapore, 1998.

\bibitem{pederson:prb00}
T.~G. Pederson, P.~M. Johansen, and H.~C. Pederson.
\newblock Particle-in-a-box model of one-dimensional excitons in conjugated
  polymers.
\newblock {\em Phys. Rev. B}, 61:10504--10510, 2000.

\bibitem{rubio:prl99}
A.~Rubio, D.~S\'anchez-Portal, E.~Artacho, P.~Ordej\'on, and J.~M. Soler.
\newblock Electronic states in a finite carbon nanotube: A one-dimensional
  quantum box.
\newblock {\em Phys. Rev. Lett.}, 82:3520--3523, 1999.

\bibitem{harrison_05}
P.~Harrison.
\newblock {\em Quantum Wells, Wires and Dots: Theoretical and Computational
  Physics of Semiconductor Nanostructures, Second Edition}.
\newblock Wiley, 2005.

\bibitem{hu:prl99}
B.~Hu, B.~Li, J.~Liu, and Y.~Gu.
\newblock Quantum chaos of a kicked particle in an infinite well potential.
\newblock {\em Phys. Rev. Lett.}, 82:4224, 1999.

\bibitem{majernik:jpa99}
V.~Majernik, R.~Charvot, and E.~Majernikova.
\newblock The momentum entropy of the infinite potential well.
\newblock {\em J. Phys. A: Math. Gen.}, 32:2207, 1999.

\bibitem{majernik:jpa97}
V.~Majernik and L.~Richterek.
\newblock Entropic uncertainty relations for the infinite well.
\newblock {\em J. Phys. A: Math. Gen.}, 30:L49--L54, 1997.

\bibitem{sanchezruiz:pla97}
J.~S\'anchez-Ruiz.
\newblock Asymptotic formula for the quantum entropy of position in energy
  eigenstates.
\newblock {\em Phys. Lett. A}, 226:7, 1997.

\bibitem{sanchezruiz:jpa99}
J.~S\'anchez-Ruiz.
\newblock Asymptotic formulae for the quantum {R}enyi entropies of position:
  application to the infinite well.
\newblock {\em J. Phys. A: Math. Gen.}, 32:3419--3432, 1999.

\bibitem{dym_72}
H.~Dym and H.~P. McKean.
\newblock {\em Fourier Series and Integrals}.
\newblock Academic Press, New York, 1972.

\bibitem{catalan:pre02}
R.~G. Catalan, J.~Garay, and R.~L\'opez-Ruiz.
\newblock Features of the extension of a statistical measure of complexity to
  continuous systems.
\newblock {\em Phys. Rev. E}, 66:011102, 2002.

\bibitem{romera_pla08}
E.~Romera and A.~Nagy.
\newblock Fisher-{R}\'enyi entropy product and information plane.
\newblock {\em Phys. Lett. A}, 372:6823--6825, 2008.

\bibitem{lopezruiz:osid09}
R.~L\'opez-Ruiz and J.~Sa{\~n}udo.
\newblock Complexity invariance by replication in the quantum square well.
\newblock {\em Open Syst. Inf. Dynamics}, 16, 2009.
\newblock 423--427.

\bibitem{nagy:pla09}
A.~{}Nagy, K.~D. Sen, and H.~E.~Montgomery Jr.
\newblock {LMC} complexity for the ground state of different quantum systems.
\newblock {\em Phys. Lett. A}, 373:2552--2555, 2009.

\bibitem{sen_11}
K.~D. Sen.
\newblock {\em Statistical Measures: Applications to Electronic Structure}.
\newblock Springer Verlag, Berlin, 2011.

\bibitem{abramowitz_64}
M.~Abramowitz and I.~A.~Stegun (Eds.).
\newblock {\em Handbook of Mathematical Functions with Formulas, Graphs, and
  Mathematical Tables}.
\newblock National Bureau of Standars, U.S. Goverment Printing Office,
  Washington D.C., 1964.

\bibitem{zozor_pa08}
S.~Zozor, M.~Portesi, and C.~Vignat.
\newblock Some extensions of the uncertainty principle.
\newblock {\em Physica A}, 387:19--20, 2008.

\bibitem{maassen_prl88}
H.~Maassen and J.~B.~M. Uffink.
\newblock Generalized entropic uncertainty relations.
\newblock {\em Phys. Rev. Lett.}, 60:1103--1106, 1988.

\bibitem{rajagopal_pla95}
A.K. Rajagopal.
\newblock The {S}obolev inequality and the {T}sallis entropic uncertainty
  relation.
\newblock {\em Phys. Lett. A}, 205:32--36, 1995.

\end{thebibliography}
\end{document}